\documentclass[journal,onecolumn,12pt]{IEEEtran}

\usepackage{times,amsmath,amsthm,epsfig,amssymb,graphicx,amsfonts}
\usepackage{hyphenat}
\usepackage{psfrag}
\usepackage{ifpdf}
\usepackage{cite}
\usepackage{array}
\usepackage{multirow}
\usepackage{graphicx}
\usepackage{color}
\usepackage{subfigure}
\usepackage{epsfig}
\usepackage{citesort}
\usepackage{setspace}
\usepackage{latexsym}
\usepackage{amssymb}
\usepackage{amsmath}
\usepackage{flexisym}

\title{{The DoF of Two-way Butterfly Networks}}

\author{Mehdi Ashraphijuo, Vaneet Aggarwal, and Xiaodong Wang \thanks{M. Ashraphijuo and X. Wang are with the Electrical Engineering Department, Columbia University, New York, NY 10027 (e-mail: mehdi@ee.columbia.edu, wangx@ee.columbia.edu). V. Aggarwal is with the School of Industrial Engineering, Purdue University, West Lafayette, IN 47907 (e-mail: vaneet@purdue.edu).}}

\newtheorem{theorem}{Theorem}

\newtheorem{corollary}{Corollary}

\newtheorem{remark}{Remark}

\begin{document}
\maketitle

\begin{abstract}
This paper studies the two-way butterfly network, a class of two-way four-unicast networks. We first show that bidirectional links do not increase the degrees of freedom for this network thus giving the first example for networks, to the best of our knowledge, where bidirectional links do not increase the degrees of freedom. Further, we see that sufficient caching at the relays or increasing the number of antennas in the relays can double the two-way degrees of freedom for butterfly network.
\end{abstract}


\newpage

\section{Introduction}

While characterizing network capacity is in general unsolved, there has been considerable progress in two research fronts. The first one focuses on single-flow multi-hop networks, in which one source aims to send the same message to one or more destinations, using multiple relay nodes. Since in this scenario all destination nodes are interested in the same message, there is effectively only one information stream in the network. Starting from the max-flow-min-cut theorem of Ford-Fulkerson \cite{ford1956maximal}, there has been significant progress on this problem \cite{avestimehr2011wireless}. The second research direction focuses on multi-flow wireless networks with only one-hop between the sources and the destinations, i.e., the interference channel. While the capacity of the interference channel remains unknown, there has been a variety of approximate capacity results, such as constant-gap capacity approximations \cite{etkin2008gaussian,bresler2010approximate,bresler2008two} and degrees of freedom characterizations \cite{cadambe2008interference,etkin2006degrees,jafar2008degrees,motahari2014real,maddah2008communication,jafar2007degrees}.

The two-way communication between two users was first studied by Shannon \cite{Shannon}. Recently, there have been many attempts to demonstrate two-way communications experimentally \cite{Chen:1998aa,Khandani,Bliss:2007,Radunovic:2009aa,Aryafar,tech_report,Bharadia}. The two-way relay channel where two users communicate with each other in the presence of relays, has been widely studied \cite{b1,b2,i4,i5,nam2010capacity,i7,i8,Rankov,s2,s5,s8,s16,s20,s22,Avestim-t2,s4,s12,s13,JJ}. Two-unicast channels consist of two sources and two destinations communicating through a general network. Degrees of freedom for one-way $2\times2\times2$ fully-connected two-unicast channels has been studied in \cite{gou2012aligned}, and further extended with interfering relays in  \cite{gou2011aligned}. These results were further generalized to one-way $2\times2\times2$ non-layered topology in \cite{gou2014toward,gou2011degrees2}. General one-way two-unicast channel has been considered in \cite{wang2011multiple,shomorony2013two} and it was shown in \cite{shomorony2013two} that the DoF for any topology takes one of the values in $\{1,\frac{3}{2},2\}$, depending on the topology. Two-way two-unicast  channels have been studied for a single relay in \cite{wang2013degrees,wang2014beyond,xin2011coordinated}. In \cite{lee2013achievable}, three different achievability strategies for two-way MIMO $2\times2\times2$ fully-connected channel are proposed. A finite-field two-way two-unicast model is also studied in \cite{maier2013cyclic,hong2013two}.

In this paper, we first study the two-way butterfly network, a class of two-way four-unicast networks and find its degrees of freedom for the cases of no caching  at the relays. Butterfly network is motivated from the network coding example \cite{Yeung2010}. This network has one-way degrees of freedom of $2$ \cite{shomorony2013two}. In this paper, we show that the two-way degrees of freedom is also $2$. This is the first result to the best of our knowledge, where bidirectional links do not improve the degrees of freedom. In order to show this result, a genie-aided outer bound is derived. This result explains the challenge in considering general two-way networks since there are network configurations where the degree of freedom double \cite{JJ}, and there are configurations where there is no benefit of bidirectional links in terms of degrees of freedom.

We further consider the case where relays in the two-way butterfly network have access to a caching memory. Caching is a technique to reduce traffic load by exploiting the high degree of asynchronous content reuse and the fact that storage is cheap and ubiquitous in today's wireless devices \cite{2cc,3cc}. During off-peak periods when network resources are abundant, some content can be stored at the wireless edge (e.g., access points or end user devices), so that demands can be met with reduced access latencies and bandwidth requirements. The caching problem has a long history, dating back to the work by Belady in 1966 \cite{1ma}. There are various forms of caching, i.e., to store data at user ends, relays, etc. \cite{4cc}. However, using the uncoded data on devices can result in an inefficient use of the aggregate cache capacity \cite{5cc}. The caching problem consists of a placement phase which is performed offline and an online delivery phase. One important aspect of this problem is the design of the placement phase in order to facilitate the delivery phase. There are several recent works that consider communication scenarios where user nodes have pre-cached information from a fixed library of possible files during the offline phase, in order to minimize the transmission from source during the delivery phase \cite{maddah2014fundamental,ji2015throughput}. There are only a limited number of works on the degrees of freedom with caching. In particular, \cite{han2015degrees,han2015improving} study the degrees of freedom for the relay and interference channels with caching, respectively, under some assumptions and provide asymptotic results on the degrees of freedom as a function of the output of some optimization problems. In this paper we show that caching increases the degrees of freedom of the butterfly network to $4$. This is the first example, to the best of our knowledge, where caching at relays increases the degrees of freedom of the two-way unicast networks with relays. It demonstrates that caching at relays doubles the degrees of freedom of two-way butterfly network. 

The remainder of this paper is as following. In Section \ref{sec2}, we present the channel model, and in Section \ref{thm_2_LB_Secq} and Section \ref{thm_2_LB_Secq2w}, we provide the main results on the DoF of proposed model without and with caching, respectively. In Section \ref{kjhsfd}, a butterfly network with a multiple-antenna relay is studied. Finally, Section \ref{sec5} concludes this paper.

\section{Channel Model}\label{sec2}

Fig. \ref{fig:3.1q} and Fig. \ref{fig:3.2q} represent the one-way and two-way butterfly networks, respectively.  As shown in Fig. \ref{fig:3.2q}, the two-way butterfly network consists of four transmitters $S_1,\dots,S_4$, three relays $R_1,\dots,R_3$, and four receivers $D_1,\dots,D_4$. Each transmitter $S_i$ has one message that is intended for its respective receiver $D_i$. Fig. \ref{fig:2q} shows the two hops of this system separately. In the first hop (Fig. \ref{fig:2.1q}), the signals received at relays in time slot $m$ are
\begin{eqnarray}
Y_{R_1}[m]&=&H_{1,R_1}X_1[m]+H_{4,R_1}X_4[m]+Z_{R_1}[m],\\
Y_{R_2}[m]&=&\sum_{i=1}^{4}H_{i,R_2}X_i[m]+Z_{R_2}[m],\\
Y_{R_3}[m]&=&H_{2,R_3}X_2[m]+H_{3,R_3}X_3[m]+Z_{R_3}[m],
\end{eqnarray}
where $H_{i,R_k}$ is the channel coefficient from transmitter $S_i$ to relay $R_k$,  $X_i[m]$ is the signal transmitted from $S_i$, $Y_{R_k}[m]$ is the signal received at relay $R_k$ and $Z_{R_k}[m]$ is the i.i.d. circularly symmetric complex Gaussian noise with zero mean and unit variance, $i\in\{1,2,3,4\}$, $k\in\{1,2,3\}$. In the second hop (Fig. \ref{fig:2.2q}), the signals received at receivers in time slot $m$ are given by
\begin{eqnarray}
Y_i[m]&=&\sum_{k=1}^{2}H_{R_k,i}X_{R_k}[m]+Z_i[m],\ \ \ \ \ \ \ {\text{for}} \ i\in \{2,3\},\\
Y_i[m]&=&\sum_{k=2}^{3}H_{R_k,i}X_{R_k}[m]+Z_i[m],\ \ \ \ \ \ \ {\text{for}} \ i\in \{1,4\},
\end{eqnarray}
where $H_{R_k,i}$ is the channel coefficient from relay $R_k$ to receiver $D_i$, $X_{R_k}[m]$ is the signal transmitted from $R_k$, $Y_i[m]$ is the signal received at receiver $D_i$ and $Z_i[m]$ is the i.i.d. circularly symmetric complex Gaussian noise with zero mean and unit variance, $i\in\{1,2,3,4\}$, $k\in\{1,2,3\}$. We assume that the channel coefficient values are drawn i.i.d. from a continuous distribution and they are bounded from above and below, i.e., $H_{\min} < | H_{i,R_k}[m] | < H_{\max}$ and $H_{\min} < | H_{R_k,i}[m] | < H_{\max}$ as in \cite{cadambe2008interference}. The relays are assumed to be full-duplex and equipped with caches. Furthermore, the relays are assumed to be causal, which means that the signals transmitted from the relays depend only on the signals received in the past and not on the current received signals and can be described as
\begin{equation}
X_{R_k}[m] = f(Y_{R_k}^{m-1},X_{R_k}^{m-1},C_{R_k}),
\end{equation}
where $X_{R_k}^{m-1}\triangleq(X_{R_k}[1],\dots,X_{R_k}[m-1])$, $Y_{R_k}^{m-1}\triangleq(Y_{R_k}[1],\dots,Y_{R_k}[m-1])$, and $C_{R_k}$ is the cached information in relay ${R_k}$. We assume that source $S_i$, $i\in\{1,2,3,4\}$ only knows channels $H_{i,R_k}$, $k\in\{1,2,3\}$; relay $R_k$, $k\in\{1,2,3\}$ only knows channels $H_{i,R_k}$ and $H_{R_k,i}$, $i\in\{1,2,3,4\}$; and destination $D_i$, $i\in\{1,2,3,4\}$ only knows channels $H_{R_k,i}$, $k\in\{1,2,3\}$.

The source $S_i$, $i\in\{1,2,3,4\}$ has a message $W_i$ that is intended for destination $D_i$. $|W_i|$ denotes the size of the message $W_i$. The rates ${\mathcal R}_i=\frac{\log |W_i|}{n}$, $i\in\{1,2,3,4\}$ are achievable during $n$ channel uses by choosing $n$ large enough, if  the probability of error can be arbitrarily small for all four messages simultaneously. The capacity region ${\mathcal C}=\{({\mathcal R}_1,{\mathcal R}_2,{\mathcal R}_3,{\mathcal R}_4)\}$ represents the set of all achievable quadruples. The sum-capacity is the maximum sum-rate that is achievable, i.e., ${\mathcal C}_{\Sigma}(P)=\sum_{i=1}^{4}{\mathcal R}^c_i$ where $({\mathcal R}_1^c,\dots,{\mathcal R}_4^c) = \arg\max _{({\mathcal R}_1,\dots,{\mathcal R}_4)\in{\mathcal C}}\sum_{i=1}^{4} {\mathcal R}_i $ and $P$ is the transmit power at each node (both source nodes and relay nodes). The degrees of freedom is defined as
\begin{equation}
DoF \triangleq \lim_{P\rightarrow\infty}\frac{{\mathcal C}_{\Sigma}(P)}{\log P}=
\sum_{i=1}^{4}\lim_{P\rightarrow\infty}\frac{{\mathcal R}^c_i}{\log P}=\sum_{i=1}^{4}d_i,
\end{equation}
where $d_i \triangleq \lim_{P\rightarrow\infty}\frac{{\mathcal R}^c_i}{\log P}$ is defined as the ${\text{DoF}}$ of source $S_i$, for $i\in\{1,2,3,4\}$. We note that ${\text{DoF}}$ is the degrees of freedom for almost every channel realization (in other words, with probability 1 over the channel realizations). We denote ${\text{DoF}}_C$ as the degrees of freedom for the case of with relay caching, and ${\text{DoF}}_{NC}$ as the degrees of freedom for the case of no relay caching.

\begin{figure}[htbp]
\centering
\subfigure[One-way butterfly network.]{
	\includegraphics[width=9cm]{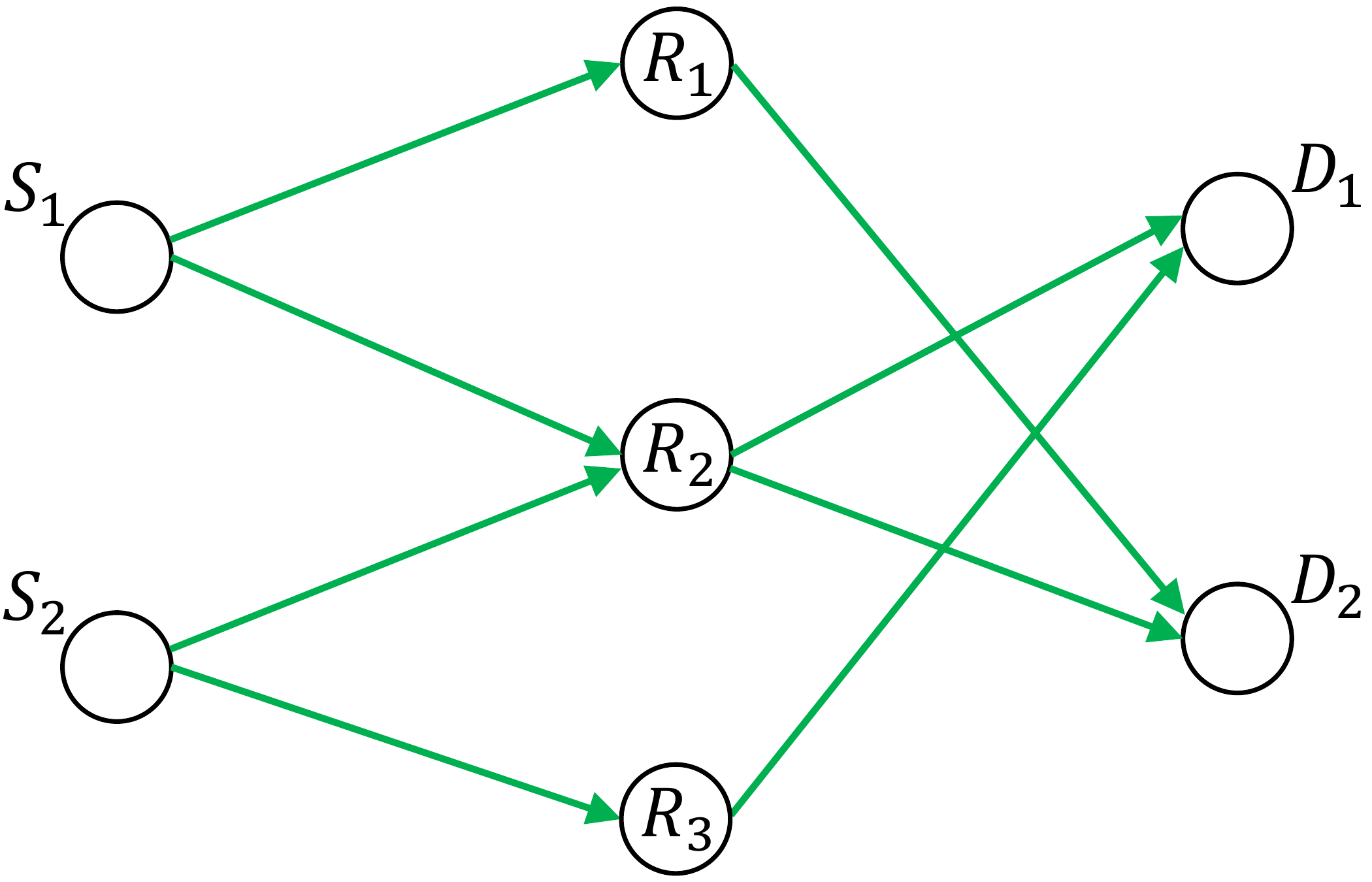}
    \label{fig:3.1q}
}
\subfigure[Two-way butterfly network.]{
	\includegraphics[width=9cm]{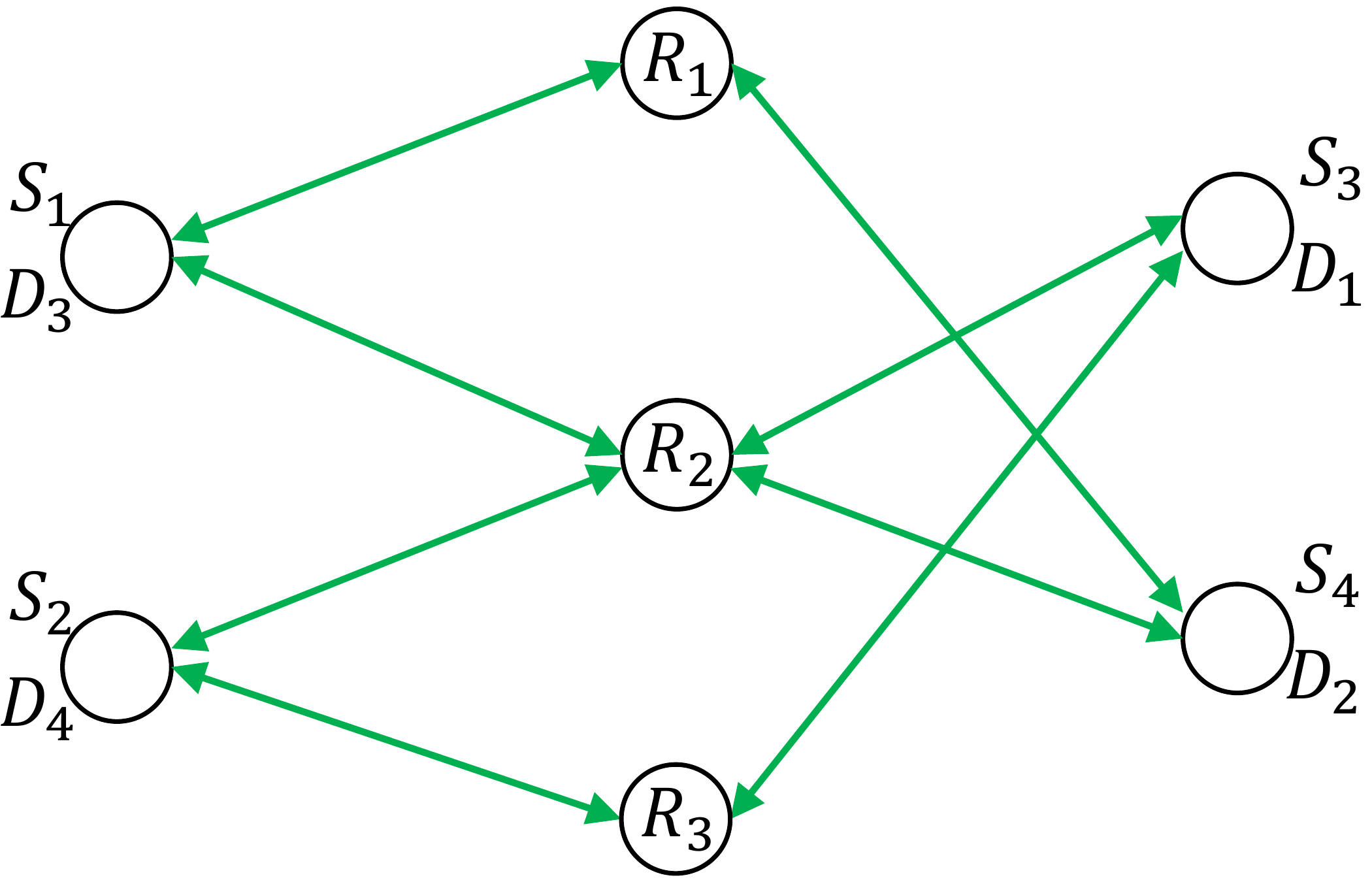}
    \label{fig:3.2q}
}
\caption[Optional caption for list of figures]{Butterfly network.}
\label{fig:3}
\end{figure}

\begin{figure}[htbp]
\centering
\subfigure[The channels from transmitters to the relays.]{
	\includegraphics[width=10cm]{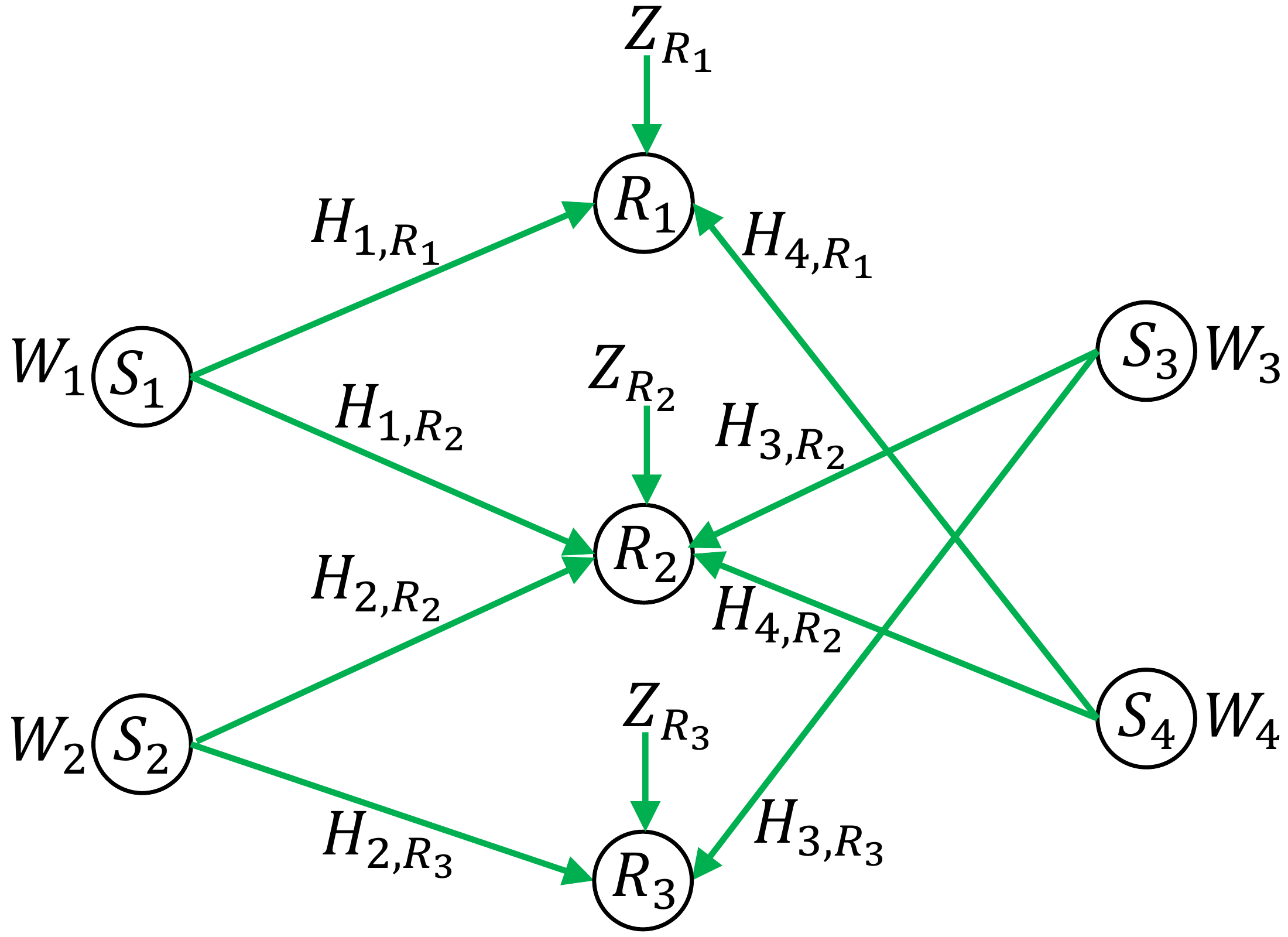}
    \label{fig:2.1q}
}
\subfigure[The channels from relays to the receivers.]{
	\includegraphics[width=10cm]{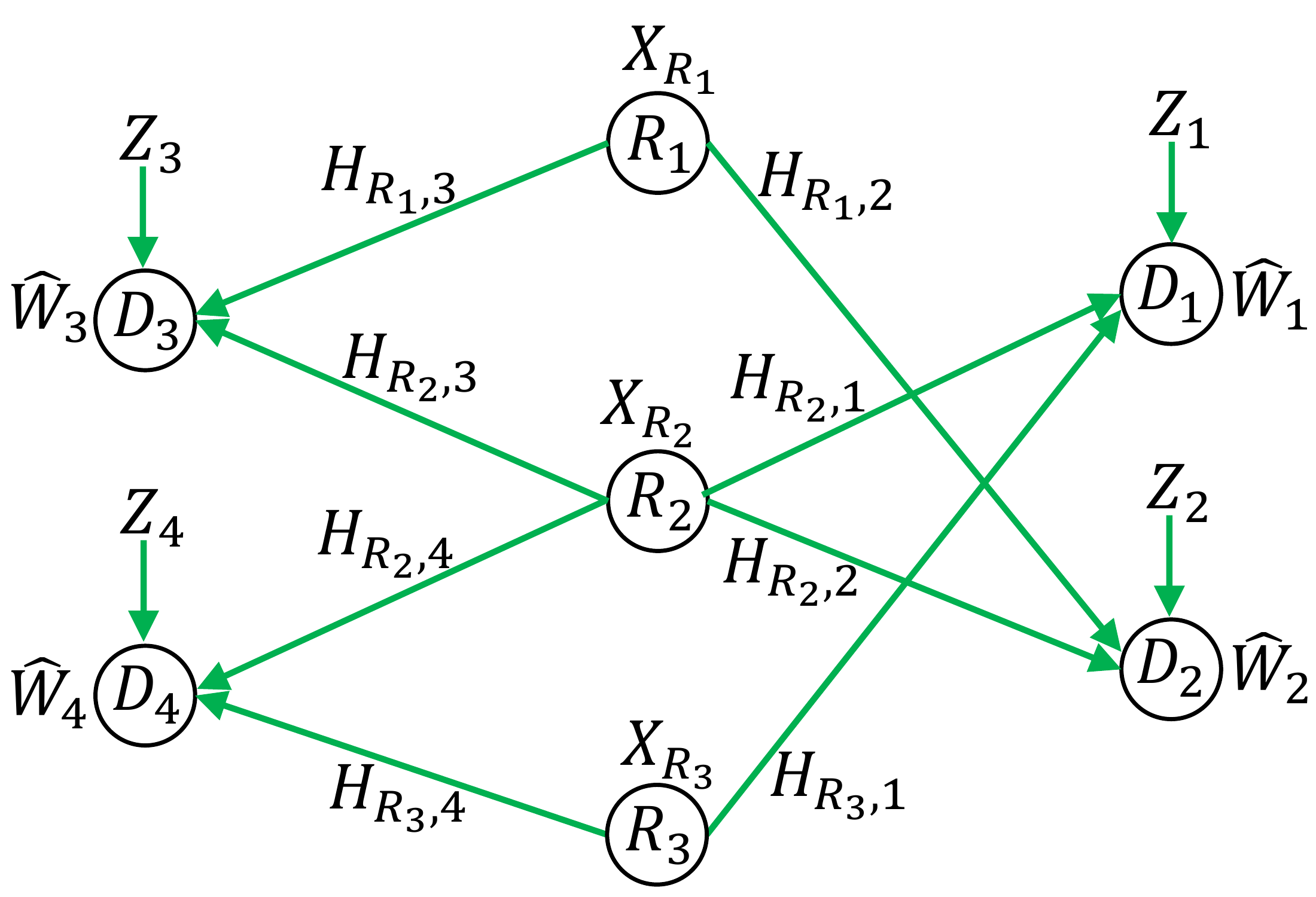}
    \label{fig:2.2q}
}
\caption[Optional caption for list of figures]{The channels from and to relays in a two-way butterfly network.}
\label{fig:2q}
\end{figure}

\section{Two-way Butterfly Network without Relay Caching}\label{thm_2_LB_Secq}

The result in this section depicts the impact of bidirectional links for the case of no relay caching. We find that the degrees of freedom for the two-way butterfly network is $2$. The result is surprising since the degrees of freedom for the one-way butterfly network given in Fig. \ref{fig:3.1q} is proven to be $2$ \cite[Theorem~1, Part B$'$]{shomorony2013two}, thus showing no improvement in the sum DoF by using two-way capabilities. To the best of our knowledge, this is the first network where two-way network achieves the same total DoF as the similar one-way network.

\begin{theorem}\label{thm_1_LBddd}
For the two-way butterfly network,  ${\text{DoF}}_{NC} = 2$.
\end{theorem}
\begin{proof}
We first show the upper bound. Consider $S_1$, $R_1$, and $S_4$ as one group of nodes and $S_2$, $R_3$, and $S_3$ as another group. Using genie-aided side information, assume that the nodes in each group have access to all of the messages in the same group. Note that the first group has $W_1$ and $W_4$ needed by the second group and the second group has $W_2$ and $W_3$ needed by the first group. The genie-aided side information does not give the needed message to any destination, and the two groups can only communicate through $R_2$. The described channel can be seen in Fig. \ref{fig:2qq}. Nodes $A_1$ and $A_2$ in the figure each have three antennas. Thus, cutset bound gives that sum DoF $\le 2$. The reason is that $R_2$ is a single antenna node and each of $A_1$ and $A_2$ can only decode one DoF of information from it.

\begin{figure}[htbp]
\centering
\subfigure[The channels from transmitters to the relays.]{
	\includegraphics[width=10cm]{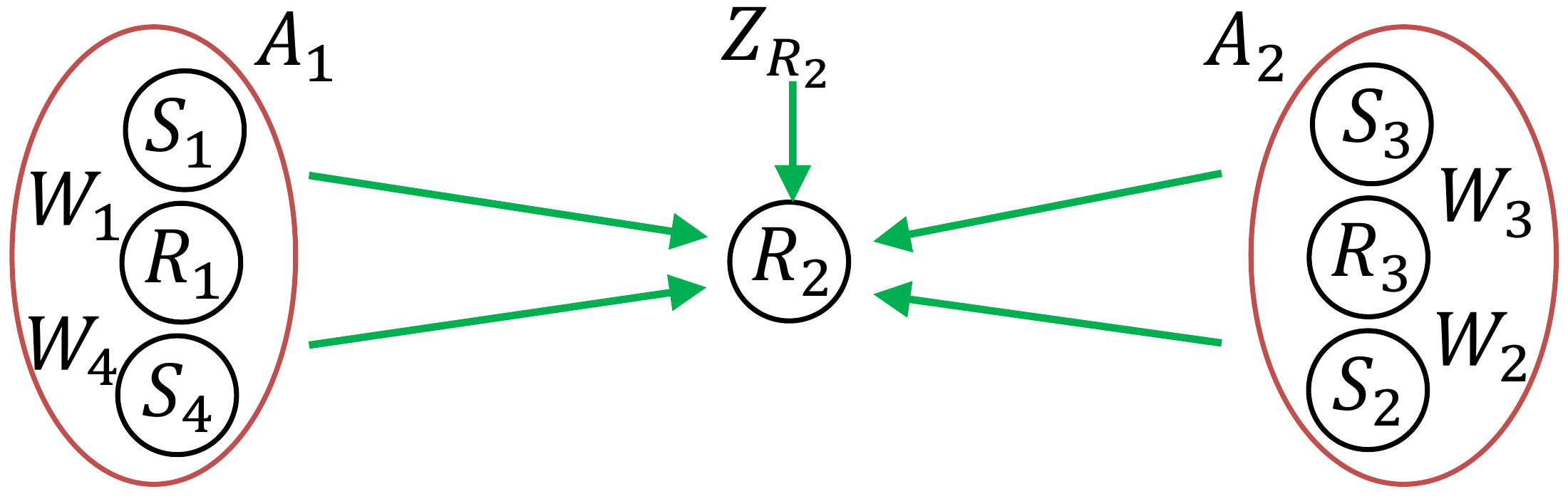}
    \label{fig:2.1qx}
}
\subfigure[The channels from relays to the receivers.]{
	\includegraphics[width=10cm]{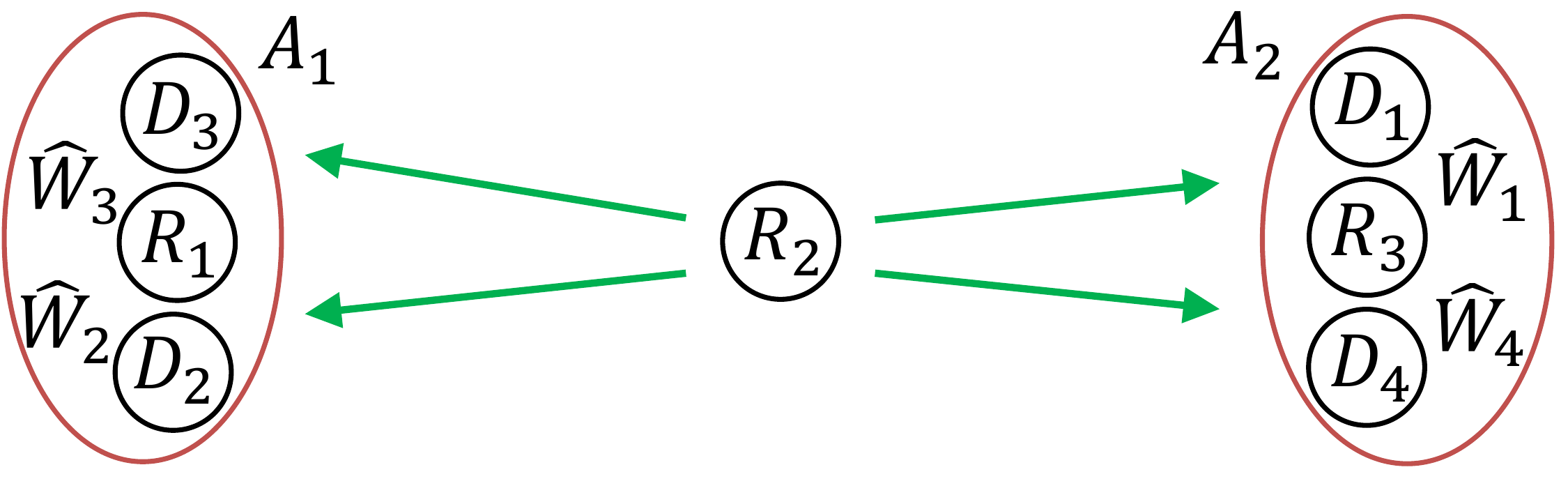}
    \label{fig:2.2qx}
}
\caption[Optional caption for list of figures]{The genie-aided butterfly network.}
\label{fig:2qq}
\end{figure}

The proof for achievability is straightforward. If all the nodes except for $S_1$, $R_2$, and $S_3$ in Fig. \ref{fig:3} are silent, then the channel can be seen as a two-way $1\times1\times1$ relay network formed by $S_1$, $R_2$, and $S_3$. This channel can achieve two degrees of freedom by simply forwarding the sum of the received signals at relay $R_2$, which is the sum of the two messages from $S_1$ and $S_3$.
\end{proof}

\begin{remark}
For the one-way butterfly network given in Fig. \ref{fig:3.1q}, the lower bound of $2$ on the total degrees of freedom is tricky  \cite[Theorem~1, Part B$'$]{shomorony2013two}, but the outer bound of $2$ is a simple cutset bound. In contrast, for the two-way butterfly network given in Fig. \ref{fig:3.2q}, as we saw in the proof of Theorem \ref{thm_1_LBddd}, the outer bound of $2$ on the total degrees of freedom is not easy to see, but the lower bound of $2$ is straightforward.
\end{remark}

\section{Two-way Butterfly Network with Relay Caching}\label{thm_2_LB_Secq2w}

\subsection{Transmission and Caching Strategy}

We now assume that each relay is equipped with a cache that can store the data from the sources. Our goal is to design strategies for caching and transmission so that the sum rate of all four source-destination pairs is maximized. Our strategy comprises two parts. The first phase is the transmission from sources to the relays, as shown in Fig. \ref{fig:2.1q}, which is performed offline and is known as the placement phase. The second phase is the transmission from relays to the destinations, as shown in Fig. \ref{fig:2.2q}, which is performed online and is known as the delivery phase. We assume that the relays decode ${W}_i$, $i=1,\dots,4$ in the offline phase and save $W'_1\triangleq {W}_1\oplus {W}_3$, $W'_2\triangleq {W}_2\oplus {W}_4$ in their caches. The transmitted signals from the relays intend to make $W'_1$ decodable at $D_1$ and $D_3$, and $W'_2$ decodable at $D_2$ and $D_4$ in Fig. \ref{fig:2.2q}. 

\subsection{Main Results on the Two-way Butterfly Network with Relay Caching}

In this section, we consider the case where the relays have a caching memory as discussed above. In this case, we show that the degrees of freedom is 4, thus depicting that the two-way degrees of freedom is twice that of the one-way degrees of freedom.

\begin{theorem}\label{thm_1_LBdc}
For the two-way butterfly network with relay caching as descried above, ${\text{DoF}}_{C} = 4$.
\end{theorem}
\begin{proof}
The upper bound follows from the cutset bound. We now provide an achievability strategy. The relays know the new messages $W'_1$ and $W'_2$ and the encoded signals in all of the relays for messages $W'_1$ and $W'_2$ at time $m\in [1,n]$ are the same, i.e., $A[m] = f(W'_1)$ and $B[m] = f(W'_2)$, respectively. At time $m$, the relays transmit the following messages
\begin{eqnarray}
X_{R_1}[m] &=&  - \frac{H_{R_2,2}}{H_{R_1,2}} A[m] - \frac{H_{R_2,3}}{H_{R_1,3}} B[m],\nonumber\\
X_{R_2}[m] &=& A[m] + B[m],\nonumber\\
X_{R_3}[m] &=& - \frac{H_{R_2,4}}{H_{R_3,4}} A[m] - \frac{H_{R_2,1}}{H_{R_3,1}} B[m].\nonumber
\end{eqnarray}
Using this, we see that the received signals at the destinations are as follows
\begin{eqnarray}
Y_1[m] &=& H_{R_2,1}\left(A[m]+ B[m]\right)+
H_{R_3,1}\left(- \frac{H_{R_2,4}}{H_{R_3,4}} A[m] - \frac{H_{R_2,1}}{H_{R_3,1}} B[m]\right) + Z_1[m]\nonumber\\
&=& \left(H_{R_2,1}- \frac{H_{R_3,1}H_{R_2,4}}{H_{R_3,4}}\right)A[m] + Z_1[m],\nonumber\\
Y_2[m] &=& H_{R_2,2}\left(A[m] + B[m]\right)+
H_{R_1,2}\left(- \frac{H_{R_2,2}}{H_{R_1,2}} A[m] - \frac{H_{R_2,3}}{H_{R_1,3}} B[m]\right) + Z_2[m]\nonumber\\
&=& \left(H_{R_2,2}- \frac{H_{R_1,2}H_{R_2,3}}{H_{R_1,3}}\right)B[m] + Z_2[m],\nonumber\\
Y_3[m] &=& H_{R_2,3}\left(A[m] + B[m]\right)+
H_{R_1,3}\left(- \frac{H_{R_2,2}}{H_{R_1,2}} A[m] - \frac{H_{R_2,3}}{H_{R_1,3}} B[m]\right) + Z_3[m]\nonumber\\
&=& \left(H_{R_2,3}- \frac{H_{R_1,3}H_{R_2,2}}{H_{R_1,2}}\right)A[m] + Z_3[m],\nonumber\\
Y_4[m] &=& H_{R_2,4}\left(A[m] + B[m]\right)+
H_{R_3,4}\left(- \frac{H_{R_2,4}}{H_{R_3,4}} A[m] - \frac{H_{R_2,1}}{H_{R_3,1}} B[m]\right) + Z_4[m]\nonumber\\
&=& \left(H_{R_2,4}- \frac{H_{R_3,4}H_{R_2,1}}{H_{R_3,1}}\right)B[m] + Z_4[m],\nonumber
\end{eqnarray}
Note that the first and the third receivers receive noisy versions of $A[m]$, from which they can decode $W'_1$ and subtract the contribution of their transmission to get the interference-free message and thus the desired signal can be decoded. The argument is similar for the second and the fourth receivers using $B[m]$ and $W'_2$ and thus showing that four degrees of freedom can be achieved.
\end{proof}

This theorem shows that for the model in Fig. \ref{fig:3}, although the bidirectional links does not increase the DoF, relay caching can achieve the maximum possible degrees of freedom, i.e., $4$.

In addition,  the following statement for the lower bound of the model with limited caching in the relays can be obtained.

{\begin{corollary}\label{thm_1_LBdcx}
For the two-way butterfly network given in Fig. \ref{fig:3} with caching $p$ portion of each of the messages $W'_1$ and $W'_2$ into each of the relays ($0\le p\le 1$), the total DoF of $2+2p$ is achievable.
\end{corollary}
\begin{proof}
We apply time-sharing to obtain this result. In $(1-p)$ portion of the time, we do not use the caching of the relays and DoF of $2$ is achievable as in Theorem \ref{thm_1_LBddd}. In the other $p$ portion of the time, we assume that the relays have access to messages $W'_1$ and $W'_2$ available from the caching and apply the same transmission strategy as in Theorem \ref{thm_1_LBdc}. So, each user can achieve the DoF of $1$ and the total DoF $= 4$ is achievable in this part of the time-sharing. So, in total time, the average sum DoF of $4p + 2(1-p) = 2 + 2p$ is achievable.
\end{proof}}

\section{Two-way Butterfly Network with Multiple-Antenna Relay}\label{kjhsfd}

In the previous section, we showed that relay caching can increasing the degrees of freedom for butterfly network. Here we see that increasing the number of antennas at relay $R_2$ can also increase the degrees of freedom to $4$ for butterfly network. 

Fig. \ref{fig:10} represents the two-way butterfly network with 3 antennas at relay $R_2$. There are some differences in the model compared with the one in previous sections as described in the following. The channels ${\bf H}_{i,R_2}$, $\forall i \in \{1,\dots,4\}$ are $3\times 1$, and the channels ${\bf H}_{R_2,i}$, $\forall i \in \{1,\dots,4\}$ are $1\times 3$ vectors. The scalar beamformers $v_1$ and $v_3$ are for transmission from relays $R_1$ and $R_3$, respectively. Also, ${\bf v}_2$ is the $3\times 3$ beamforming matrix for transmission from relay $R_2$. In addition, ${\bf Y}_{R_2}[m]$ is the $3\times 1$ vector signal received at relay $R_2$ and ${\bf Z}_{R_2}[m]$ is the $3\times 1$ vector i.i.d. circularly symmetric complex Gaussian additive noise with zero mean and unit variance entries at relay $R_2$.

\begin{figure}[htbp]
\centering
	\includegraphics[width=8.5cm]{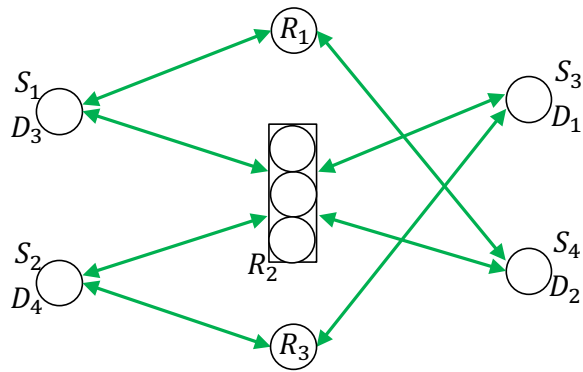}
\caption{Two-way butterfly network with $3$ antennas at $R_2$.}
\label{fig:10}
\end{figure}

The following theorem shows that increasing the number of antennas at relays can increase the degrees of freedom of channel. 

\begin{theorem}\label{thm_mul}
For the two-way butterfly network with $3$ antennas at relay $R_2$ as descried above, ${\text{DoF}} = 4$.
\end{theorem}
\begin{proof}
The upper bound follows from the cutset bound. We now provide an achievability strategy. The received signals at relays are as following:
\begin{eqnarray}
Y_{R_1}[m] &=& H_{1,R_1}X_1[m] + H_{4,R_1}X_4[m] + Z_{R_1}[m],\label{1eq}\\
{\bf Y}_{R_2}[m] &=& \sum_{i=1}^{4} {\bf H}_{i,R_2}X_i[m] + {\bf Z}_{R_2}[m],\label{2eq}\\
Y_{R_3}[m] &=& H_{2,R_3}X_2[m] + H_{3,R_3}X_3[m] + Z_{R_3}[m].\label{3eq}
\end{eqnarray}
And given the received signals and the relay beamformers, we can write the received signals in the destinations as below:
\begin{eqnarray}
Y_1[m] &=& H_{R_1,1} v_1 Y_{R_1}[m]+ {\bf H}_{R_2,1} {\bf v_2} {\bf Y}_{R_2}[m] + Z_1[m],\label{1qq}\\
Y_2[m] &=& H_{R_3,2} v_3 Y_{R_3}[m]+ {\bf H}_{R_2,2} {\bf v_2} {\bf Y}_{R_2}[m] + Z_2[m],\label{2qq}\\
Y_3[m] &=& H_{R_3,3} v_3 Y_{R_3}[m]+ {\bf H}_{R_2,3} {\bf v_2} {\bf Y}_{R_2}[m] + Z_3[m],\label{3qq}\\
Y_4[m] &=& H_{R_1,4} v_1 Y_{R_1}[m]+ {\bf H}_{R_2,4} {\bf v_2} {\bf Y}_{R_2}[m] + Z_4[m].\label{4qq}
\end{eqnarray}

Applying \eqref{1eq}-\eqref{3eq} into \eqref{1qq}-\eqref{4qq} gives the following:
\begin{eqnarray}
Y_1[m] &=& H_{R_1,1} v_1 \left(H_{1,R_1}X_1[m] + H_{4,R_1}X_4[m] + Z_{R_1}[m]\right)+\nonumber\\
&& {\bf H}_{R_2,1} {\bf v_2} \left(\sum_{i=1}^{4} {\bf H}_{i,R_2}X_i[m] + {\bf Z}_{R_2}[m]\right) + Z_1[m],\label{1qqt}\\
Y_2[m] &=& H_{R_3,2} v_3 \left(H_{2,R_3}X_2[m] + H_{3,R_3}X_3[m] + Z_{R_3}[m]\right)+\nonumber\\
&& {\bf H}_{R_2,2} {\bf v_2} \left(\sum_{i=1}^{4} {\bf H}_{i,R_2}X_i[m] + {\bf Z}_{R_2}[m]\right) + Z_2[m],\label{2qqt}\\
Y_3[m] &=& H_{R_3,3} v_3 \left(H_{2,R_3}X_2[m] + H_{3,R_3}X_3[m] + Z_{R_3}[m]\right)+\nonumber\\
&& {\bf H}_{R_2,3} {\bf v_2} \left(\sum_{i=1}^{4} {\bf H}_{i,R_2}X_i[m] + {\bf Z}_{R_2}[m]\right) + Z_3[m],\label{3qqt}\\
Y_4[m] &=& H_{R_1,4} v_1 \left(H_{1,R_1}X_1[m] + H_{4,R_1}X_4[m] + Z_{R_1}[m]\right)+\nonumber\\
&& {\bf H}_{R_2,4} {\bf v_2} \left(\sum_{i=1}^{4} {\bf H}_{i,R_2}X_i[m] + {\bf Z}_{R_2}[m]\right) + Z_4[m].\label{4qqt}
\end{eqnarray}

If we find a set of beamformers in relays such that intended signals can be decoded in destinations, the proof of achievability is complete. For example, in destination $D_1$, $X_1[m]$ is the intended signal and the destination already has access to $X_3[m]$. Therefore, the multipliers of $X_2[m]$ and $X_4[m]$ should be zero in destination $D_1$. Therefore, we conclude the following two equations from \eqref{1qqt}:
\begin{eqnarray}
0 &=& {\bf H}_{R_2,1} {\bf v_2}  {\bf H}_{2,R_2},\label{1q}\\
0 &=& H_{R_1,1} v_1 H_{4,R_1} + {\bf H}_{R_2,1} {\bf v_2}  {\bf H}_{4,R_2}.\label{2q}
\end{eqnarray}
Similarly, the multipliers of $X_2[m]$ and $X_4[m]$ should be zero in destination $D_3$, and using \eqref{3qqt}, the following two equations are obtained:
\begin{eqnarray}
0 &=& {\bf H}_{R_2,3} {\bf v_2}  {\bf H}_{4,R_2},\label{3q}\\
0 &=& H_{R_3,3} v_3 H_{2,R_3} + {\bf H}_{R_2,3} {\bf v_2}  {\bf H}_{2,R_2}.\label{4q}
\end{eqnarray}
Similarly, the multipliers of $X_1[m]$ and $X_3[m]$ should be zero in destination $D_2$, and using \eqref{2qqt}, the following two equations are obtained:
\begin{eqnarray}
0 &=& {\bf H}_{R_2,2} {\bf v_2}  {\bf H}_{1,R_2},\label{5q}\\
0 &=& H_{R_3,2} v_3 H_{3,R_3} + {\bf H}_{R_2,2} {\bf v_2}  {\bf H}_{3,R_2}.\label{6q}
\end{eqnarray}
Similarly, the multipliers of $X_1[m]$ and $X_3[m]$ should be zero in destination $D_4$, and using \eqref{4qqt}, the following two equations are obtained:
\begin{eqnarray}
0 &=& {\bf H}_{R_2,4} {\bf v_2}  {\bf H}_{3,R_2},\label{7q}\\
0 &=& H_{R_1,4} v_1 H_{1,R_1} + {\bf H}_{R_2,4} {\bf v_2}  {\bf H}_{1,R_2},\label{8q}
\end{eqnarray}
In our proof, we assume that $v_1 = v_3 = 0$. So, we do not take advantage of relays $R_1$ and $R_3$. So, Eqns. \eqref{1q}-\eqref{8q} reduce to

\begin{eqnarray}
0 &=& {\bf H}_{R_2,1} {\bf v_2}  {\bf H}_{2,R_2},\label{1mq}\\
0 &=& {\bf H}_{R_2,1} {\bf v_2}  {\bf H}_{4,R_2},\label{2mq}\\
0 &=& {\bf H}_{R_2,3} {\bf v_2}  {\bf H}_{4,R_2},\label{3mq}\\
0 &=& {\bf H}_{R_2,3} {\bf v_2}  {\bf H}_{2,R_2},\label{4mq}\\
0 &=& {\bf H}_{R_2,2} {\bf v_2}  {\bf H}_{1,R_2},\label{5mq}\\
0 &=& {\bf H}_{R_2,2} {\bf v_2}  {\bf H}_{3,R_2},\label{6mq}\\
0 &=& {\bf H}_{R_2,4} {\bf v_2}  {\bf H}_{3,R_2},\label{7mq}\\
0 &=& {\bf H}_{R_2,4} {\bf v_2}  {\bf H}_{1,R_2}.\label{8mq}.
\end{eqnarray}
We need to find a non-zero solution for this set of equations to complete the proof. This is easy to see, because there are $9$ parameters (entries of ${\bf v_2}$) and $8$ Eqns. \eqref{1mq}-\eqref{8mq}. So, the solution to Eqns. \eqref{1mq}-\eqref{8mq} is at least a one-dimension space. This means there are infinitely many non-zero solutions to Eqns. \eqref{1mq}-\eqref{8mq}. 
\end{proof}

\begin{corollary}
In butterfly network given in Fig. \ref{fig:10}, if we remove relays $R_1$ and $R_3$, the degrees of freedom is still $4$.
\end{corollary}
\begin{proof}
The proof follows from the proof of Theorem \ref{thm_mul}. It is easy to see that the outer bound still holds. Also, in achievability scheme of Theorem \ref{thm_mul} relays $R_1$ and $R_3$ are not used.
\end{proof}

We note that relays $R_1$ and $R_3$ did not help to achieve optimal degrees of freedom for $1$ antenna case at $R_2$ (Theorem \ref{thm_1_LBddd}), and $3$ antennas case at $R_2$ (Theorem \ref{thm_mul}). Whether these relays help for two antennas case at $R_2$ is an open problem. 

\begin{remark}
We considered butterfly network with three antennas at $R_2$ as in Fig. \ref{fig:10}. However, with two antennas at $R_2$, above strategy do not directly work. More specifically, in the proof of Theorem \ref{thm_mul} with two antennas at $R_2$, there are still $8$ interference removal equations, but ${\bf v_2}$ will be $2\times 2$ and together with $v_1$ and $v_3$ it leads to only $6$ parameters. 
\end{remark}

\section{Conclusions}\label{sec5}
This paper studied the two-way butterfly network, a class of two-way four-unicast networks. We showed that bidirectional links do not increase the degrees of freedom for this network. Further, it was shown that enough caching at the relays or increasing the number of antennas in relays double the degrees of freedom for this network.

The considered class of two-way four-unicast networks, i.e., butterfly networks, demonstrates the challenges of finding degrees of freedom for a general single antenna network. Thus, finding the degrees of freedom for the general two-way four-unicast networks, with or without caching, remains an open problem. Finally, the DoF of butterfly network with multiple-antennas at each node is still open.

\bibliographystyle{IEEETran}
\bibliography{bib}

\begin{thebibliography}{10}
\providecommand{\url}[1]{#1}
\csname url@samestyle\endcsname
\providecommand{\newblock}{\relax}
\providecommand{\bibinfo}[2]{#2}
\providecommand{\BIBentrySTDinterwordspacing}{\spaceskip=0pt\relax}
\providecommand{\BIBentryALTinterwordstretchfactor}{4}
\providecommand{\BIBentryALTinterwordspacing}{\spaceskip=\fontdimen2\font plus
\BIBentryALTinterwordstretchfactor\fontdimen3\font minus
  \fontdimen4\font\relax}
\providecommand{\BIBforeignlanguage}[2]{{%
\expandafter\ifx\csname l@#1\endcsname\relax
\typeout{** WARNING: IEEEtran.bst: No hyphenation pattern has been}%
\typeout{** loaded for the language `#1'. Using the pattern for}%
\typeout{** the default language instead.}%
\else
\language=\csname l@#1\endcsname
\fi
#2}}
\providecommand{\BIBdecl}{\relax}
\BIBdecl

\bibitem{ford1956maximal}
L.~R. Ford and D.~R. Fulkerson, ``Maximal flow through a network,''
  \emph{Canadian journal of Mathematics}, vol.~8, no.~3, pp. 399--404, 1956.

\bibitem{avestimehr2011wireless}
A.~S. Avestimehr, S.~N. Diggavi, and D.~N. Tse, ``Wireless network information
  flow: A deterministic approach,'' \emph{IEEE Transactions on Information
  Theory}, vol.~57, no.~4, pp. 1872--1905, 2011.

\bibitem{etkin2008gaussian}
R.~H. Etkin, D.~N. Tse, and H.~Wang, ``Gaussian interference channel capacity
  to within one bit,'' \emph{IEEE Transactions on Information Theory}, vol.~54,
  no.~12, pp. 5534--5562, 2008.

\bibitem{bresler2010approximate}
G.~Bresler, A.~Parekh, and D.~N. Tse, ``The approximate capacity of the
  many-to-one and one-to-many {G}aussian interference channels,'' \emph{IEEE
  Transactions on Information Theory}, vol.~56, no.~9, pp. 4566--4592, 2010.

\bibitem{bresler2008two}
G.~Bresler and D.~Tse, ``The two-user {G}aussian interference channel: a
  deterministic view,'' \emph{European transactions on telecommunications},
  vol.~19, no.~4, pp. 333--354, 2008.

\bibitem{cadambe2008interference}
V.~R. Cadambe and S.~A. Jafar, ``Interference alignment and degrees of freedom
  of the $k$-user interference channel,'' \emph{IEEE Transactions on
  Information Theory}, vol.~54, no.~8, pp. 3425--3441, 2008.

\bibitem{etkin2006degrees}
R.~H. Etkin and D.~N. Tse, ``Degrees of freedom in some underspread {MIMO}
  fading channels,'' \emph{IEEE Transactions on Information Theory}, vol.~52,
  no.~4, pp. 1576--1608, 2006.

\bibitem{jafar2008degrees}
S.~A. Jafar and S.~Shamai, ``Degrees of freedom region of the {MIMO} {X}
  channel,'' \emph{IEEE Transactions on Information Theory}, vol.~54, no.~1,
  pp. 151--170, 2008.

\bibitem{motahari2014real}
A.~S. Motahari, S.~Oveis-Gharan, M.-A. Maddah-Ali, and A.~K. Khandani, ``Real
  interference alignment: Exploiting the potential of single antenna systems,''
  \emph{IEEE Transactions on Information Theory}, vol.~60, no.~8, pp.
  4799--4810, 2014.

\bibitem{maddah2008communication}
M.~A. Maddah-Ali, A.~S. Motahari, and A.~K. Khandani, ``Communication over
  {MIMO} {X} channels: Interference alignment, decomposition, and performance
  analysis,'' \emph{IEEE Transactions on Information Theory}, vol.~54, no.~8,
  pp. 3457--3470, 2008.

\bibitem{jafar2007degrees}
S.~A. Jafar and M.~J. Fakhereddin, ``Degrees of freedom for the {MIMO}
  interference channel,'' \emph{IEEE Transactions on Information Theory},
  vol.~53, no.~7, pp. 2637--2642, 2007.

\bibitem{Shannon}
C.~Shannon, ``Two-way communication channels,'' in \emph{Proc. 4th Berkeley
  Symp. Mathematical Statistics Probability}, 1961, pp. 611--644.

\bibitem{Chen:1998aa}
S.~Chen, M.~Beach, and J.~McGeehan, ``Division-free duplex for wireless
  applications,'' \emph{Electronics Letters}, vol.~34, no.~2, pp. 147--148, Jan
  1998.

\bibitem{Khandani}
A.~K. Khandani, ``Methods for spatial multiplexing of wireless two-way
  channels,'' US patent, filed Oct. 2006 (provisional patent filed Oct. 2005),
  issued Oct. 2010.

\bibitem{Bliss:2007}
D.~W. Bliss, P.~Parker, and A.~R. Margetts, ``Simultaneous transmission and
  reception for improved wireless network performance,'' in \emph{Proc. of
  IEEE/SP 14th Workshop on Statistical Signal Processing}, August 2007, pp.
  478--482.

\bibitem{Radunovic:2009aa}
B.~Radunovic, D.~Gunawardena, P.~Key, A.~P.~N. Singh, V.~Balan, and G.~Dejean,
  ``Rethinking indoor wireless: Low power, low frequency, full duplex,'' in
  \emph{Fifth IEEE Workshop on Wireless Mesh Networks}, 2010.

\bibitem{Aryafar}
E.~Aryafar, M.~A. Khojastepour, K.~Sundaresan, S.~Rangarajan, and M.~Chiang,
  ``{MIDU: Enabling MIMO Full Duplex},'' in \emph{{Proc. of the ACM Mobicom}},
  2012.

\bibitem{tech_report}
M.~Duarte, A.~Sabharwal, V.~Aggarwal, R.~Jana, K.~K. Ramakrishnan, C.~Rice, and
  N.~K. Shankaranayanan, ``Design and characterization of a full-duplex
  multiantenna system for {WiFi} networks,'' \emph{IEEE Transactions on
  Vehicular Technology}, vol.~63, no.~3, pp. 1160--1177, March 2014.

\bibitem{Bharadia}
D.~Bharadia, E.~Mcmilin, and S.~Katti, ``Full duplex radios,'' in \emph{ACM
  SIGCOMM}, Oct. 2013, pp. 375--386.

\bibitem{b1}
T.~Oechtering, C.~Schnurr, I.~Bjelakovic, and H.~Boche, ``Broadcast capacity
  region of two-phase bidirectional relaying,'' \emph{IEEE Transactions on
  Information Theory}, vol.~54, no.~1, pp. 454--458, Jan. 2008.

\bibitem{b2}
G.~Kramer and S.~Shamai, ``Capacity for classes of broadcast channels with
  receiver side information,'' in \emph{Proc. IEEE Information Theory Workshop
  (ITW)}, Lake Tahoe, CA, Sep. 2007.

\bibitem{i4}
S.~Kim, P.~Mitran, and V.~Tarokh, ``Performance bounds for bidirectional coded
  cooperation protocols,'' \emph{IEEE Transactions on Information Theory},
  vol.~54, no.~11, pp. 5235--5241, Nov. 2008.

\bibitem{i5}
S.~Kim, N.~Devroye, P.~Mitran, and V.~Tarokh, ``Achievable rate regions and
  performance comparison of half duplex bi-directional relaying protocols,''
  \emph{IEEE Transactions on Information Theory}, vol.~57, no.~10, pp.
  6405--6418, Oct. 2011.

\bibitem{nam2010capacity}
W.~Nam, S.-Y. Chung, and Y.~H. Lee, ``Capacity of the gaussian two-way relay
  channel to within $\frac{1}{2}$ bit,'' \emph{IEEE Transactions on Information
  Theory}, vol.~56, no.~11, pp. 5488--5494, Nov. 2010.

\bibitem{i7}
M.~Wilson, K.~Narayanan, H.~Pfister, and A.~Sprintson, ``Joint physical layer
  coding and network coding for bidirectional relaying,'' \emph{IEEE
  Transactions on Information Theory}, vol.~56, no.~11, pp. 5641--5654, Nov.
  2010.

\bibitem{i8}
B.~Nazer and M.~Gastpar, ``Compute-and-forward: Harnessing interference through
  structured codes,'' \emph{IEEE Transactions on Information Theory}, vol.~57,
  no.~10, p. 6463–6486, Oct. 2011.

\bibitem{Rankov}
B.~Rankov and A.~Wittneben, ``Achievable rate regions for the two-way relay
  channel,'' in \emph{Proc. IEEE International Symposium on Information Theory
  Proceedings (ISIT)}, Seattle, WA, Jul. 2006, pp. 1668--1672.

\bibitem{s2}
S.~Ghasemi-Goojani and H.~Behroozi, ``Nested lattice codes for gaussian two-way
  relay channels,'' \emph{arXiv preprint arXiv:1301.6291}, Jan. 2013.

\bibitem{s5}
V.~Havary-Nassab, S.~Shahbazpanahi, and A.~Grami, ``Optimal distributed
  beamforming for two-way relay networks,'' \emph{IEEE Transactions on Signal
  Processing}, vol.~58, no.~3, pp. 1238--1250, Mar. 2010.

\bibitem{s8}
B.~Jiang, F.~Gao, X.~Gao, and A.~Nallanathan, ``Channel estimation and training
  design for two-way relay networks with power allocation,'' \emph{IEEE
  Transactions on Wireless Communications}, vol.~9, no.~6, pp. 2022--2032, Jun.
  2010.

\bibitem{s16}
I.~K, P.~V, S.~Bhashyam, and A.~Thangaraj, ``Outer bounds for the capacity
  region of a {G}aussian two-way relay channel,'' in \emph{Proc. 50th Annual
  Allerton Conference on Communication, Control, and Computing (Allerton)},
  Urbana, IL, Oct. 2012.

\bibitem{s20}
Y.~Song, N.~Devroye, H.-R. Shao, and C.~Ngo, ``Lattice coding for the two-way
  two-relay channel,'' \emph{arXiv preprint arXiv:1212.1198}, Dec. 2012.

\bibitem{s22}
H.~Yang, J.~Chun, and A.~Paulraj, ``Asymptotic capacity of the separated {MIMO}
  two-way relay channel,'' \emph{IEEE Transactions on Information Theory},
  vol.~57, no.~11, pp. 7542--7554, Nov. 2011.

\bibitem{Avestim-t2}
A.~Avestimehr, A.~Sezgin, and D.~Tse, ``Capacity of the two-way relay channel
  within a constant gap,'' \emph{European Transactions on Telecommunications},
  vol.~21, no.~4, pp. 363--374, Jun. 2010.

\bibitem{s4}
D.~Gunduz, E.~Tuncel, and J.~Nayak, ``Rate regions for the separated two-way
  relay channel,'' in \emph{Proc. 46th Annual Allerton Conference on
  Communication, Control, and Computing (Allerton)}, Urbana, IL, Sep. 2008.

\bibitem{s12}
E.~Yilmaz and R.~Knopp, ``Hash-and-forward relaying for two-way relay
  channel,'' in \emph{Proc. IEEE International Symposium on Information Theory
  Proceedings (ISIT)}, Saint-Petersburg, Russia, Jul.-Aug. 2011.

\bibitem{s13}
L.~Ong and S.~Johnson, ``The capacity region of the restricted two-way relay
  channel with any deterministic uplink,'' \emph{IEEE Transactions on
  Information Theory}, vol.~16, no.~3, pp. 396--399, Mar. 2012.

\bibitem{JJ}
M.~Ashraphijuo, V.~Aggarwal, and X.~Wang, ``On the capacity regions of two-way
  diamond channels,'' \emph{IEEE Transactions on Information Theory}, vol.~60,
  no.~11, pp. 6060--6090, Nov. 2015.

\bibitem{gou2012aligned}
T.~Gou, S.~A. Jafar, C.~Wang, S.-W. Jeon, and S.-Y. Chung, ``Aligned
  interference neutralization and the degrees of freedom of the
  2$\times$2$\times$2 interference channel,'' \emph{IEEE Transactions on
  Information Theory}, vol.~58, no.~7, pp. 4381--4395, Jul. 2012.

\bibitem{gou2011aligned}
T.~Gou, C.~Wang, S.~Jafar \emph{et~al.}, ``Aligned interference neutralization
  and the degrees of freedom of the 2$\times$ 2$\times$ 2 interference channel
  with interfering relays,'' in \emph{Proc. 49th Annual Allerton Conference on
  Communication, Control, and Computing (Allerton)}, Urbana, IL, Sep. 2011, pp.
  1041--1047.

\bibitem{gou2014toward}
------, ``Toward full-duplex multihop multiflow -- a study of non-layered two
  unicast wireless networks,'' \emph{IEEE Journal on Selected Areas in
  Communications}, vol.~32, no.~9, pp. 1738--1751, Sep. 2014.

\bibitem{gou2011degrees2}
------, ``Degrees of freedom of a class of non-layered two unicast wireless
  networks,'' in \emph{Proc. Conference Record of the Forty Fifth Asilomar
  Conference on Signals, Systems and Computers (ASILOMAR)}, Pacific Grove, CA,
  Nov. 2011, pp. 1707--1711.

\bibitem{wang2011multiple}
C.~Wang, T.~Gou, S.~Jafar \emph{et~al.}, ``Multiple unicast capacity of
  2-source 2-sink networks,'' in \emph{Proc. IEEE Global Telecommunications
  Conference (GLOBECOM)}, Houston, TX, Dec. 2011, pp. 1--5.

\bibitem{shomorony2013two}
I.~Shomorony and S.~Avestimehr, ``Two-unicast wireless networks: Characterizing
  the degrees of freedom,'' \emph{IEEE Transactions on Information Theory},
  vol.~59, no.~1, pp. 353--383, Jan. 2013.

\bibitem{wang2013degrees}
C.~Wang and S.~A. Jafar, ``Degrees of freedom of the two-way relay {MIMO}
  interference channel,'' \emph{e-print UC-escholarship: 9qc3343h, UCI CPCC
  report}, Jan. 2013.

\bibitem{wang2014beyond}
C.~Wang, ``Beyond one-way communication: Degrees of freedom of multi-way relay
  {MIMO} interference networks,'' \emph{arXiv preprint arXiv:1401.5582}, Jan.
  2014.

\bibitem{xin2011coordinated}
H.~Xin, Y.~Peng, C.~Wang, Y.~Yang, and W.~Wang, ``Coordinated eigen beamforming
  for multi-pair {MIMO} two-way relay network,'' in \emph{Proc. IEEE Global
  Telecommunications Conference (GLOBECOM)}, Houston, TX, Dec. 2011, pp. 1--6.

\bibitem{lee2013achievable}
K.~Lee, N.~Lee, and I.~Lee, ``Achievable degrees of freedom on {MIMO} two-way
  relay interference channels,'' \emph{IEEE Transactions on Wireless
  Communications}, vol.~12, no.~4, pp. 1472--1480, Apr. 2013.

\bibitem{maier2013cyclic}
H.~Maier and R.~Mathar, ``Cyclic interference neutralization on the
  2$\times$2$\times$2 full-duplex two-way relay-interference channel,'' in
  \emph{Proc. IEEE Information Theory Workshop (ITW)}, Seville, Spain, Sep.
  2013, pp. 1--5.

\bibitem{hong2013two}
S.-N. Hong and G.~Caire, ``Two-unicast two-hop interference network:
  Finite-field model,'' in \emph{Proc. IEEE Information Theory Workshop (ITW)},
  Seville, Spain, Sep. 2013, pp. 1--5.

\bibitem{Yeung2010}
\BIBentryALTinterwordspacing
R.~W. Yeung, ``Network coding theory: An introduction,'' \emph{Frontiers of
  Electrical and Electronic Engineering in China}, vol.~5, no.~3, pp. 363--390,
  2010. [Online]. Available: \url{http://dx.doi.org/10.1007/s11460-010-0103-1}
\BIBentrySTDinterwordspacing

\bibitem{2cc}
N.~Golrezaei, A.~G. Dimakis, A.~F. Molisch, and G.~Caire, ``Wireless video
  content delivery through distributed caching and peer-to-peer gossiping,'' in
  \emph{Conference Record of the Forty Fifth Asilomar Conference on Signals,
  Systems and Computers (ASILOMAR)}, 2011, pp. 1177--1180.

\bibitem{3cc}
A.~F. Molisch, G.~Caire, D.~Ott, J.~R. Foerster, D.~Bethanabhotla, and M.~Ji,
  ``Caching eliminates the wireless bottleneck in video aware wireless
  networks,'' \emph{Advances in Electrical Engineering}, 2014.

\bibitem{1ma}
L.~A. Belady, ``A study of replacement algorithms for a virtual-storage
  computer,'' \emph{IBM Systems journal}, vol.~5, no.~2, pp. 78--101, 1966.

\bibitem{4cc}
X.~Wang, M.~Chen, T.~Taleb, A.~Ksentini, and V.~Leung, ``Cache in the air:
  exploiting content caching and delivery techniques for {5G} systems,''
  \emph{IEEE Communications Magazine}, vol.~52, no.~2, pp. 131--139, 2014.

\bibitem{5cc}
N.~Golrezaei, K.~Shanmugam, A.~G. Dimakis, A.~F. Molisch, and G.~Caire,
  ``Femtocaching: Wireless video content delivery through distributed caching
  helpers,'' in \emph{Proceedings IEEE INFOCOM}, 2012, pp. 1107--1115.

\bibitem{maddah2014fundamental}
M.~A. Maddah-Ali and U.~Niesen, ``Fundamental limits of caching,'' \emph{IEEE
  Transactions on Information Theory}, vol.~60, no.~5, pp. 2856--2867, 2014.

\bibitem{ji2015throughput}
M.~Ji, G.~Caire, and A.~F. Molisch, ``The throughput-outage tradeoff of
  wireless one-hop caching networks,'' \emph{IEEE Transactions on Information
  Theory}, vol.~61, no.~12, pp. 6833--6859, 2015.

\bibitem{han2015degrees}
W.~Han, A.~Liu, and V.~K. Lau, ``Degrees of freedom in cached {MIMO} relay
  networks,'' \emph{IEEE Transactions on Signal Processing}, vol.~63, no.~15,
  pp. 3986--3997, 2015.

\bibitem{han2015improving}
------, ``Improving the degrees of freedom in {MIMO} interference network via
  {PHY} caching,'' in \emph{IEEE Global Communications Conference (GLOBECOM)},
  2015, pp. 1--6.

\end{thebibliography}

\end{document}